\documentclass{article}
\usepackage{amsmath}

\usepackage{amsthm}
\usepackage{sectsty}
\usepackage{mathdots}
\usepackage{indentfirst}
\usepackage{tikz}
\usepackage{siunitx}
\usepackage{amssymb}
\usepackage{setspace}
\usepackage{amsfonts}
\usepackage{float}
\usepackage{graphicx}
\usepackage{subcaption}
\usepackage[export]{adjustbox}
\usepackage[margin=1in]{geometry}
\sectionfont{\centering}
\newtheorem{theorem}{Theorem}[section]

\newtheorem{proposition}[theorem]{Proposition}
\newtheorem{corollary}[theorem]{Corollary}

\newtheorem{lemma}[theorem]{Lemma}
\title{Strategically Acting on Information\footnote{I thank Prof. Abreu and Prof. Pearce for their guidance and helpful suggestions on this project.}}
\author{Xiaoming Wang}
\date{December 2024}
\begin{document}

\maketitle

\begin{abstract}
In many non-cooperative settings, agents often possess useful information that provide an advantage over their opponent(s), but acting on such information too frequently can lead to detection. I develop a simple framework to analyze such a trade-off and characterize the optimal way in which to act on information.
\end{abstract}

\section{Introduction}

The Allied success during World War II in cracking the German Enigma cipher was considered pivotal in expediting the end of the conflict and saving millions of lives. However, it wasn't just enough to decrypt German communications. Equally important was the task of determining what intelligence to act on. This is perhaps best illustrated in the movie \textit{The Imitation Game}. The following scene occurs soon after Alan Turing and his team discover that a fleet of German U-boats are about to target a British convey. While everyone else is eager to inform the Admiralty to thwart the attack, Turing wisely points out the risks of doing so:
\bigskip
\par
Alan: \textit{``You know why people like violence, Hugh? It's because it feels good. Sometimes we can't do what feels good. We have to do what is logical. Hardest time to lie to somebody is when they're expecting to be lied to. $\ldots$ So our convey suddenly veers off course. A squadron of RAF (Royal Air Force) bombers miraculously descends on the coordinates of the U-boats. What will the Germans think?"}
\smallskip
\par
Hugh: \textit{``The Germans will know we have broken Enigma.''}
\bigskip
\par
Besides war, it is not difficult to envision situations in day-to-day life where one might stand to gain by selectively acting on information. A corrupt government official who wishes to embezzle public funds, for example, would be judicious to do so at irregular intervals and in small amounts to reduce their chances of being caught. Likewise, in poker, a player might be able to tell from someone's body language whether they actually have a good hand or are just bluffing. However, it might only be worth capitalizing on this information when the stakes of the game are high enough.
\medskip
\par
More broadly, such considerations can arise even in the event where it is common knowledge that one party has the ability to inflict harm on the other. The United States, for instance, levies sanctions against other nations, entities, or individuals that it deems a national security threat. Nevertheless, it has been recognized \cite{sanctions} that overuse of such powers leads adversaries to decouple their economies from the U.S. to blunt potential negative impacts arising from geopolitical tensions. Indeed, in the ongoing debate over whether the U.S. should liquidate Russian funds to finance the Ukraine war, one of the principal arguments against such measures is that it would scare off future investments to the West. As such, it becomes vital to understand the optimal way in which to utilize information for maximal long-term gain.
\medskip
\par
To refer once again to $\textit{The Imitation Game}$: Later on in the film, in exchange for help from the chief of British Secret Service MI6 to keep their breaking of Enigma a secret, Turing pledges to
\bigskip
\par
\textit{``... develop a system for determining how much intelligence to act on. Which attacks to stop, which to let through. Statistical analysis. The minimal number of actions it would take for us to win the war, but the maximal number we can take before the Germans get suspicious.''}
\bigskip
\par
The goal of this paper is to offer an explicit interpretation of this analysis. In particular, I consider two players engaged in a (possibly) infinitely repeated game with discount factor $\delta$. Player 1 (he) is assumed to have information on player 2 (she) and can choose to either act on that information or feign ignorance in each period. To be more concrete, say that player 2 has a set of available actions to take, and player 1 is to ``guess'' what that action is. The payoff to him is $1$ if he gets it correct and $0$ if he doesn't. On the other hand, Player 2 believes player 1 is partially informed in the sense that she thinks he has a probability $p$ of guessing correctly and holds a Beta$(\alpha,\beta)$ prior over the distribution of $p$, which is updated according to Bayes rule after each period upon observing the outcome. As soon as player 2's expected belief exceeds some threshold $c$, she cuts off future interaction permanently. Otherwise, the game continues indefinitely. 
\medskip
\par
The main result is a classification of the set of optimal strategies for player 1 as a correspondence of $\delta$ for all $\alpha,\beta\in\mathbb{N}$ and a sequence of thresholds $c$ going to $0$. In what follows, Section 2 reviews some related papers and highlights some connections with a few existing strands of literature. Section 3 presents the analysis, and Section 4 considers limitations of the model and potential extensions.

\section{Literature Review}

Despite the prevalence and broad applicability of the described phenomenon, there seems to be little existing literature on this topic. The paper that inspired this project, and which comes closest to modeling this dynamic, is \cite{lee2023feigning}. The motivation there is essentially the same, but being an experimental paper, the setup is quite different. In particular, Lee considers a two-stage game with two players, a seeker and a hider, where the seeker has a predefined probability of being an informed type, meaning they can observe the hiders action before making a decision. In both stages, it is beneficial for the seeker but damaging to the hider if their actions match. After observing the seeker's action, the hider decides whether or not to continue to the second stage. Lee shows that if the payoff in the first stage is large, the informed seeker will match the hider's action in the first period, whereas if the payoff in the first stage is small, the informed seeker will mimic an uninformed seeker so that they can wait until the second stage to match.
\medskip
\par
A closely related paper that addresses the issue of spying and its implications is \cite{solan2004games}. The authors there consider one-shot games where one of two players can purchase information regarding the other's action before the game begins, and where the precision of information increases in the price the player pays. Interestingly enough, espionage creates two opposing effects: a direct effect from knowledge of opponent strategies, but also an indirect effect driven by the opponent's awareness that they are being spied upon. This allows them to commit to a particular strategy beforehand and thereby effectively gain a ``first-mover advantage''. The interplay between the two forces determines the outcome of the game.
\medskip
\par
Espionage also plays an important role in industrial organization. \cite{barrachina2014entry}, for example, studies the effects of industrial espionage on entry deterrence. They consider a monopoly incumbent who may expand capacity to deter entry, and a potential entrant who owns an intelligence system which generates a noisy signal based on the incumbents actions. The incumbent may signal-jam to manipulate the likelihood of the noisy signal and hence influence the entrant’s decisions. The authors find that when the accuracy of the intelligence system is commonly known, it is the incumbent who benefits from their rival's espionage, and the higher the precision, the more they benefit. Conversely, when the accuracy is private information to the potential entrant, a higher precision in the intelligence system benefits the entrant.
\medskip
\par
Another connection worth noting is the literature on strategic experimentation initiated by \cite{gittins1979bandit}, particularly the resemblance to the multi-armed bandit problem. In both cases, the agent's goal is to maximize expected payoff, but in the traditional setup the agent does not know the value of each arm, and needs to ``explore'' them while ``exploiting'' the gains. In contrast, the agent in my model can be seen as knowing the value of each arm but seeking to ``exploit'' them in a way that makes it look like as if he were ``exploring''.
\medskip
\par
The theme of feigning ignorance for long-term gains can also be viewed in some sense as a dual problem to games of reputation such as the chain-store paradox \cite{selten1978chain}. Recall that there, the long-run incumbent has an incentive to project strength by fighting in early periods even when he is weak so as to scare off potential entrants in the future. In my model, the informed seeker is in fact strong, but wants to display weakness so as to keep his opponent around. 

\section{Model and Analysis}

The basic setup was given in the introduction and will not be repeated. Nevertheless, a few comments on some of the assumptions and information structure are in order: In contrast to a more standard model where player 1 simply has some probability of being informed, in my view the notion of being either perfectly informed or perfectly uniformed seems unrealistic. Hence, I find it more natural to let player 2 adopt the perspective that player 1 has some chance of correctly predicting her move. As player 1 succeeds more frequently, player 2 becomes more convinced of player 1's ability to cheat, until at some point she decides it is better to opt-out.
\medskip
\par
The choice of a Beta distribution stems from its property of being the conjugate prior of the binomial distribution and having a particularly simple expression for its expectation. Recall that given a Beta$(\alpha,\beta)$ prior, Bayesian updating yields a posterior of Beta$(\alpha+x,\beta+y)$ after observing $x$ successes and $y$ failures. This makes the analysis particularly tractable, but the drawback is that player 2 cannot place any point mass on player 1 being perfectly informed.
\medskip
\par
Now onto the analysis: note that player 1 can either choose $s$ or $f$ in each period. A strategy $\textbf{x}$ for player 1 is a tuple of the form $(x_{1},\ldots,x_{n})$ or $(x_{1},\ldots,x_{n},\ldots)$, where $x_{i}$ is either $s$ or $f$. In addition, letting $n_{s,k}$ and $n_{f,k}$ denote the number of $s$ and $f$ entries respectively up to the $k^{\text{th}}$ coordinate, a strategy is said to be \textit{feasible} if
\begin{equation}\label{eq:1}
    \frac{\alpha+n_{s,k}}{(\alpha+n_{s,k})+(\beta+n_{f,k})}\leq c
\end{equation}
for all $k<n$ when $\textbf{x}=(x_{1},\ldots,x_{n})$ and for all $k\in\mathbb{N}$ when $\textbf{x}=(x_{1},\ldots,x_{n},\ldots)$. This just means that either player 1 chooses success at some point, which increases player 2's suspicion beyond the threshold $c$, thereby terminating the game, or he keeps player 2's suspicion under $c$ forever.
\medskip
\par
The following observation gives a necessary condition for a strategy to be optimal regardless of what $\delta$ is. This will narrow the scope of our search.

\begin{proposition}
\label{proposition 1}
    If $\textbf{x}$ is an optimal strategy, then changing any $x_{i}$ from $f$ to $s$ must violate condition (\ref{eq:1}) at $k=i$.
\end{proposition}

\begin{proof}
    Suppose it didn't violate condition (\ref{eq:1}) at $k=i$. Note that there must eventually be another $s$ after $x_{i}$, as otherwise $\textbf{x}$ can't be optimal. Letting $j=\inf\{k>i\,|\,x_{k}=s\}$, consider the new strategy $\overline{\textbf{x}}$ obtained from $\textbf{x}$ by interchanging $x_{i}$ and $x_{j}$. By construction, $\overline{\textbf{x}}$ is feasible because there are only failures between $i$ and $j$: if condition (\ref{eq:1}) isn't violated at $k=i$, it can't be violated at any $k<j$. Furthermore, at $k=j$, the number of successes and failures are the same for $\textbf{x}$ and $\overline{\textbf{x}}$, so condition (\ref{eq:1}) again can't be violated. However, at the same time, $\overline{\textbf{x}}$ yields an extra utility of $\delta^{i}-\delta^{j}>0$, thereby contradicting the assumption that $\textbf{x}$ is an optimal strategy.
\end{proof}

In other words, if a (feasible) strategy is to be optimal, player 1 must choose success whenever doing so does not push player 2's suspicion above the given threshold. Denote the set of strategies satisfying this property as $\Gamma(\alpha,\beta,c)$.

\begin{corollary}
    For any $\textbf{x},\textbf{y}\in\Gamma(\alpha,\beta,c)$ such that $\textbf{x}\neq\textbf{y}$, it must be that $|\textbf{x}|\neq |\textbf{y}|$, where $|\cdot|$ is the length of the tuple.
\end{corollary}

\begin{proof}
    Suppose $\textbf{x},\textbf{y}$ were such that $\textbf{x}\neq\textbf{y}$, but $|\textbf{x}|=|\textbf{y}|$. Then let $i\in\mathbb{N}$ be the smallest index at which $x_{i}\neq y_{i}$. Without loss of generality, let's assume that $x_{i}=s$ and $y_{i}=f$. Then changing $y_{i}$ from $f$ to $s$ does not violate condition (\ref{eq:1}). This contradicts the proposition.
\end{proof}

In light of this corollary, we may enumerate the elements of $\Gamma(\alpha,\beta,c)$ in increasing order according to their length: $\{h^{1},h^{2},\ldots,h^{\infty}\}$, with $h^{1}$ being the strategy of only successes up until player 2's suspicion exceeds $c$.
\medskip
\par
To get some further insight into this characterization, we rearrange equation (\ref{eq:1}) as
\begin{equation}\label{eq:2}
    n_{f,k}\geq \frac{1-c}{c}n_{s,k}+\frac{1-c}{c}\alpha-\beta.
\end{equation}

We can visualize this as follows: Imagine a snake starting from the origin. Moving one step to the right corresponds to a success while moving one step up corresponds to a failure. Note that by equation (\ref{eq:2}), a strategy is feasible if the snake stays in the region above the blue line. The proposition simply says that the snake must keep on moving right until it hits or is about to hit the blue line, at which point it can either take one final step and cross the blue line or move up. In Figure \ref{fig:graph}, the strategy corresponding to the red line $(s,s,s,f,s,f,s,s)$ could potentially be an optimal strategy, whereas the strategy corresponding to the green line $(f,s,s,f,s,s,f,s,s,s)$ cannot.
\medskip
\par
For our purposes, we will simply find the optimal strategies for all $c=1/(m+1)$, where $m\in\mathbb{N}$. Solutions for other values of $c$ follow the same flavor and adds little insight at the expense of significantly more tedious algebra. 
\medskip
\par
First, we explicitly describe $h^{2}$, which turns out to be important in the final characterization of optimal strategies.

\begin{lemma}
\label{lemma 1}
    Given any $\alpha,\beta\in\mathbb{N}$ and $c=1/(m+1)$ such that $c\leq\alpha/(\alpha+\beta)$, write $\beta=mr+k$, where $k<m$. Then $h^{2}$ consists of $r-\alpha$ successes, followed by $m-k$ failures, followed by $2$ successes.
\end{lemma}

\begin{proof}
    To see this, observe that the $n_{s,k}$ intercept
    \begin{equation*}
        \frac{\beta}{m}-\alpha=r-\alpha+\frac{k}{m}
    \end{equation*}
    has an integer part of $r-\alpha$, so this is the number of initial successes. Now let $l$ be the the minimum number of failures that must occur before another success does not violate condition $(1)$. Then by definition,
    \begin{equation*}
        m(r-\alpha+1)+m\alpha-\beta=m-k\leq l,
    \end{equation*}
    meaning that $h^{2}$ is followed by $m-k$ failures. Finally, $h^{2}$ must then contain only successes until condition (\ref{eq:2}) is violated. Since $m\geq 1$, this equates to $2$ successes.
\end{proof}

\begin{figure}
\centering
\begin{tikzpicture}[scale=1]
    % Draw grid lines
    \draw[help lines, color=gray!30] (-1,-1) grid (7,6);
    % Draw x-axis and y-axis
    \draw[thick,->] (-1,0) -- (7,0) node[right] {$n_{s,k}$};
    \draw[thick,->] (0,-1) -- (0,6) node[above] {$n_{f,k}$};
    % Draw integer points
    \foreach \x in {0,1,2,...,6} {\foreach \y in {0,1,2,...,5} {\fill (\x,\y) circle (0.05);}}
    % Draw the line y=x-3
    \draw[blue, thick] plot[domain=2:7] (\x,\x-3);
    % Red candidate strategy
    \draw[red, ultra thick] (0,0) -- (3,0);
    \draw[red, ultra thick] (3,0) -- (3,1);
    \draw[red, ultra thick] (3,1) -- (4,1);
    \draw[red, ultra thick] (4,1) -- (4,2);
    \draw[red, ultra thick] (4,2) -- (5,2);
    \draw[red, ultra thick] (5,2) -- (6,2);
    % Green non-candidate solution
    \draw[green, ultra thick] (0,0) -- (0,1);
    \draw[green, ultra thick] (0,1) -- (2,1);
    \draw[green, ultra thick] (2,1) -- (2,2);
    \draw[green, ultra thick] (2,2) -- (4,2);
    \draw[green, ultra thick] (4,2) -- (4,3);
    \draw[green, ultra thick] (4,3) -- (6,3);
    \draw[green, ultra thick] (6,3) -- (6,5);
\end{tikzpicture}
\caption{$\alpha=1, \beta=3, c=\frac{1}{2}$}
\label{fig:graph}
\end{figure}

Now we can state the main result:

\begin{theorem}
    For $c=1/(m+1)$, $m\in\mathbb{N}$, the optimal strategy of player 1 is
    \begin{equation*}
        \begin{cases}
            h^{1}, & \text{ if } \delta<z_{m-k};\\
            \{h^{1},h^{2}\}, & \text{ if } \delta=z_{m-k};\\
            h^{2}, & \text{ if } z_{n-k}<\delta<z_{m};\\
            \{h^{2},h^{3},\ldots,h^{\infty}\}, & \text{ if } \delta=z_{m};\\
            h^{\infty}, & \text{ if } \delta>z_{m}.
        \end{cases}
    \end{equation*}
\end{theorem}

\begin{proof}
    Note that while the set $\Gamma(\alpha,\beta,c)$ is infinite, the incentives governing the agent's decision between $h^{i}$ and $h^{i+1}$ are the same for every $i\geq 2$: eliminating the common discount factor, if another success will drive player 2's suspicion above $c$, player 1 can either do just that and end it with a success, or he could wait $m$ more rounds for two consecutive successes. Thus, if $1<\delta^{m}+\delta^{m+1}$, i.e. $\delta>z_{m}$, $h^{i+1}$ yields a higher payoff than $h^{i}$. Since this applies to every $i\geq 2$, he will simply choose to keep player 2 engaged forever. When $1<\delta^{m}+\delta^{m+1}$, we must consider the trade-offs between $h^{1}$ and $h^{2}$. Note that the number of failures that player 1 must endure to get two successes in this case is $m-k$ by Lemma \ref{lemma 1}, so he prefers $h^{2}$ over $h^{1}$ when $1<\delta^{m-k}+\delta^{m-k+1}$, and vice versa.
\end{proof}

\begin{proposition}
\label{proposition 2}
    For every $n\in\mathbb{N}$, the equation $1=x^{n}+x^{n+1}$ has a unique solution $z_{n}$ in $[0,1]$, and $z_{n}$ is increasing in $n\in\mathbb{N}$.
\end{proposition}

\begin{proof}
    Let $f_{n}(x)=x^{n}+x^{n+1}-1$. For existence, observe that $f_{n}(0)=-1<0$ while $f_{n}(1)=1>0$, so by the intermediate value theorem, a solution exists. For uniqueness, notice that $f_{n}'(x)=nx^{n-1}+(n+1)x^{n}>0$ for all $x\in(0,1)$, so $f(x)$ is monotonic. To see that $z_{n}$ is increasing in $n\in\mathbb{N}$, observe that $f_{n}(x)$ is decreasing in $n$ for all $x\in(0,1)$.
\end{proof}

Part of this result is entirely expected: when the discount factor is high, player 1 is patient enough that he is never willing to let player 2 terminate the game for any immediate gains, and when the discount factor is low, player 1 simply wishes to grab as much as possible in the short-run. Furthermore, by proposition \ref{proposition 2}, as player 2's threshold $c$ decreases, it becomes increasingly unprofitable for player 1 to feign ignorance because of the longer wait needed to reap rewards. What is perhaps a bit surprising and more non-trivial is that for intermediate values of the discount factor, it is optimal for player 1 to deliberately fail several times before clinching successes until player 2 terminates the game. The idea is that, after the first few periods, any optimal strategy always involves cycling through a sequence of failures before a success, and while it may not be profitable to go through all those failures in return for an extra success, the unique nature of the prior allows the seeker to go through fewer failures before an extra success is rendered possible.

\section{Discussion and Further Directions of Investigation}

There are several possible extensions to this model:
\medskip
\par
The first and most obvious one is to make rewards heterogeneous. Indeed, in nearly all imaginable real life scenarios, the driving incentive behind when to act on information is perhaps the value of gain in doing so relative to the expected reward per period. This could be modeled as the reward $X$ being drawn randomly in each period, with player 1 observing the draw and then deciding whether to succeed for fail. In light of such a change, proposition \ref{proposition 1} would no longer hold true: if the realized gain was small, it would be better to fail even if succeeding doesn't push player 2's suspicion above the threshold.
\medskip
\par
A second extension would be to add either a fixed or random termination mechanism by player 2. Namely, she can terminate the game if her suspicion level exceeds a threshold, but also with some exogenous probability in every period or as a Poisson process with some fixed arrival rate. Again, this phenomenon is pervasive in reality. Every year, for example, university members are asked to change their account login passwords even in the absence of a security breach. The motivation for this is presumably to protect against hackers who are lurking in the shadows waiting to strike at the right moment.
\medskip
\par
A third possible extension would be to consider the role of competition and its effects on this decision-process. Instead of one, suppose there were many players with knowledge of the same flaw in a system. If one player doesn't take advantage of the information, they risk letting someone else capture the benefits. The incentive structure in this case would be identical to a repeated prisoner's dilemma, whereby cooperating is better than not cooperating, but acting on the information while other players don't is best.
\medskip
\par
A major limitation of my model is that player 2 remains mechanical throughout and lacks higher order beliefs. Namely, she simply observes player 1's action, updates her belief and chooses whether to terminate the game based on a threshold. In this sense, it is more of an operations research type problem rather than a game theoretic one. In particular, I do not address issues of counter-intelligence: under the optimal strategy when player 1 is sufficiently patient, the pattern of several successive failures followed by a success would probably not go unnoticed and would be inconsistent with player 2's belief that her opponent is merely guessing. In other words, such strategies would not make player 1 look as if he were truly exploring. At the same time, one could argue that if player 1 feigns ignorance by deliberately missing sufficiently often, then perhaps player 2 doesn't really care so much if player 1 is actually informed or not. Therefore, it seems necessary to consider the payoffs of player 2 as well.
\medskip
\par
In an effort to address these concerns, I have attempted to extend \cite{lee2023feigning} to a (possibly) infinitely repeated game. The setup is as follows: Two players engage in a repeated game of hide-and-seek with discount factor $\delta$. There are $n$ available actions to choose from in each period for both players, and there are two types of seekers, informed and uniformed. The informed seeker can see the hider's chosen action before making his move, whereas the uninformed seeker is assumed to be behavioral, and can only randomize each period equally among all the actions. The initial probability of the seeker being uninformed is $\beta_{0}$, and is common knowledge. After observing the outcome in each period, the hider can decide whether to opt-out of the game or continue. The payoffs to the seeker and hider are, respectively, $a,-a$ if a match occurs, $0,b$ if no match occurs, and $0,0$ if the hider decides to terminate the game, where $a>b>0$.
\medskip
\par
I haven't derived any meaningful results, but some general aspects of the solution can be pointed out. First of all, the lack of a final period precludes a solution via backward induction. To make things tractable, I restrict attention to Markov perfect equilibriums, where the players' strategies depend only on the hider's belief, and not on past behavior. This seems quite reasonable in practice. In particular, let $\beta$ denote the hider's belief in the probability that the seeker is uninformed. It is fairly straightforward to see that for sufficiently high $\beta$, the informed seeker will match with certainty, and the hider will continue with certainty, because she is extremely confident that the match was due to the seeker having gotten lucky. For lower ranges of $\beta$, the hider must punish the informed seeker for transgressions by opting-out at times, so let $q(\beta')$ denote the probability of the hider exiting upon observing a match and $V(\beta')$ denote the maximum continuation value of the informed seeker when the hider continues, where $\beta'$ is the hider's posterior upon observing a match. Now if the hider is willing to randomize, it must mean that she is indifferent between the two options. Since opting-out yields a continuation value of $0$, this must also be her continuation value of staying in. To simplify things, I let the informed seeker adopt the strategy whereby he matches with probability $p(\beta)$ so that the hider's expected value is exactly $0$ in each period. Then the fact that he is willing to randomize means he must be indifferent between matching and missing. This translates into
\begin{equation*}
    (1-q(\beta'))V(\beta')=V(\beta''),
\end{equation*}
where $\beta''$ is the hider's posterior after observing a miss.

\newpage
\bibliographystyle{apalike}
\bibliography{reference}

\end{document}